%% file: active.tex
\documentclass[letterpaper,onecolumn,11pt]{article}
\usepackage{enumerate}
\usepackage{graphicx}
\usepackage[T1]{fontenc}
\usepackage{fancyvrb}

\usepackage{amsthm}
\usepackage{amsmath}
\usepackage{multicol}
\usepackage{fancyvrb}

\usepackage{epstopdf}

\usepackage{multirow}
\usepackage{hhline}

\usepackage[%
    font={small},
            format=hang,    
                format=plain,
                    margin=0pt,
                        width=0.95\textwidth,
                        ]{caption}
\usepackage[list=true,font={small}]{subcaption}

\newtheorem{definition}{Definition}
\newtheorem{mrule}{Rule}
\newtheorem{proposition}{Proposition}
\newtheorem{claim}{Claim}

\newtheorem{lemma}{Lemma}
\newcommand{\remove}[1]{}
\newcommand{\actm}{~{\small \sf activemonitor}~}
\newcommand{\lk}{{\small \sf LK}}
\newcommand{\am}{{\small \sf AM}}
\newcommand{\as}{{\small \sf AS}}
\newcommand{\abq}{{\small \sf ABQ}}
\begin{document}

\date{}
\title{ActiveMonitor: Non-blocking Monitor Executions
for Increased   
Parallelism}
\author{
Wei-Lun Hung~~~~~
Himanshu Chauhan\thanks{Contact Author. Address: ACES 5.432, 201 E 24 Street, Austin, Texas, USA}~~~~~
Vijay K. Garg\thanks{Supported in part by the NSF Grants CNS-1346245,
CNS-1115808, and Cullen Trust for Higher Education Endowed Professorship } \\
\\
The University of Texas at Austin\\
{\sf {\small \{wlhung,himanshu\}@utexas.edu~~~garg@ece.utexas.edu}}
}
\remove{
\author{
 \IEEEauthorblockN{
Wei-Lun Hung~~~~~
Himanshu Chauhan\thanks{Contact Author. Address: ACES 5.432, 201 E 24 Street, Austin, Texas, USA}~~~~~
Vijay K. Garg \thanks{supported in part by the
NSF Grants  CNS-1115808, CNS-0718990, CNS-0509024,
    and Cullen Trust for Higher Education Endowed Professorship.}
 }
 \IEEEauthorblockA{
The University of Texas at Austin
    {\small \textsf{\{wlhung,himanshu\}@utexas.edu~~~~garg@ece.utexas.edu}}
     }
}
}
\maketitle
\input{abstract.tex}
\remove{
\hfill\\
\hfill\\
\hfill\\
\hfill\\
\hfill\\
\hfill\\
\\
\vspace{50pt}
Regular paper (can be considered for brief announcement)\\
\vspace{30pt}
Eligible for the best student paper award\thispagestyle{empty}

\pagebreak
\clearpage
\setcounter{page}{1}
}
\input{introduction.tex}

\input{background.tex}
\input{concept.tex}
\input{evaluation.tex}

\input{discussion.tex}

\input{related.tex}
\input{conclusion.tex}

\bibliographystyle{abbrv}
\bibliography{ref}  

\input{appendix}
\end{document}

%% file: abstract.tex
We present a set of novel ideas on design and implementation  of monitor objects for
multi-threaded 
programs. Our approach has two main goals: (a) increase
parallelism in monitor objects and thus provide performance gains (shorter runtimes) for 
multi-threaded programs, and (b) introduce constructs that allow programmers to 
easily write monitor-based multi-threaded programs that can achieve these performance
gains. We describe the concepts of our framework, called {\em ActiveMonitor}, and its prototype implementation 
using {\em futures}\cite{future}. 
We evaluate its performance in terms of 
runtimes of multi-threaded programs 
on linked-list, bounded-buffer, and other
fundamental problems implemented in Java. We compare 
the runtimes of our implementation against 
implementations using Java's reentrant locks \cite{lea05}, recently 
proposed automatic signaling framework 
AutoSynch \cite{hg13}, and some other techniques from the literature.    
The results of 
of the evaluation indicate that monitors based on  
our framework provide significant gains in runtime performance in comparison to traditional 
monitors implemented using Java's reentrant locks. 

%% file: introduction.tex
\section{Introduction} \label{sec:intro}
\input{intro_himanshu.tex}

%% file: intro_himanshu.tex
Monitors are the prevalent programming technique for 
synchronization between threads in shared-memory parallel programs. 
They were designed \cite{hoa74} with two primary 
goals: $(i)$ to ensure safety of shared data by enforcing 
mutual exclusion in critical sections, and $(ii)$ to enable conditional synchronization using the {\em wait-notify} mechanism. 
Monitors ensure that while a thread is executing a critical
section, other threads that
cannot access the critical section {\em wait} for it to become available. 
Waiting threads are notified by some thread exiting the critical section.
We argue that this monitor design has certain drawbacks:
\remove{from performing updates at any instance. For example, consider the fundamental problem of a 
bounded-buffer shared between multiple producer and consumer threads. Fig.~\ref{fig:bb_exp}(a)
shows a Java based implementation in which locks\cite{lea05} are used to implement the 
prevalent monitor design for this problem. 
The {\sf \small put} and {\sf \small take} methods form two critical sections in this 
program. 
We argue that this implementation does not provide good throughput due to some
additional drawbacks of the current design of monitors: }
\begin{enumerate}[($a$)] \itemsep0em
\item For executions on multi-core CPUs, a thread's update to the shared data may cause a cache invalidation for some other thread
running on a different core, specially if the thread was {\em waiting} for the critical 
section. Thus, the penalty incurred due to cache misses would slow down the program on multi-core CPUs. 
\item If threads do not require the result of their updates on shared data immediately after performing these updates,  
then the blocking executions of critical sections for these updates
may become a bottleneck, and could reduce the program throughput. For example, consider a thread 
that inserts items in a linked-list. If this thread does not need the result of the insert operation
then blocking insert operations lead to slower runtimes.  
\remove{
If the items are actually being generated by some other source, and 
the producer thread is merely responsible for receiving them and inserting them
into the buffer, then 
this implementation does not provide good throughput under the heavy load 
of item generation. This is because the {\sf \small put} method becomes a bottleneck
due to its blocking execution. With the emerging trends 
of big data and sensor data analytics, it is important to consider such a use case. }
\end{enumerate}
\remove{ A common solution to tackle the throughput issue is to use reader/writer (R/W) locks; but they can only be used in
use-cases where read only monitor operations are present and also frequent. Even with R/W locks, it is 
possible that executions of a particular monitor write method may become the bottleneck, and thus result in reduced throughput. }

In this paper, we focus 
on monitors methods that perform updates on shared data. Thus,
the solution of using Read/Write locks \cite{lea05} 
is not applicable. 
Observe that the design of monitors to enforce blocking executions 
was envisioned in $1970$'s  
when processors were single-core, processing 
power was scarce and saving processor
cycles was a primary concern for programmers. Under these settings, making the 
critical section executions blocking was not a performance issue. The current situation is drastically different:  
not only multi-core processors are now ubiquitous,  
but they are  
also significantly cheaper and faster. Therefore, it is important to explore
alternate design and implementation of monitors to exploit the availability of multi-core processors. Given 
that the blocking executions of monitor methods seem to form a bottleneck for performance, we 
pose the following question: can monitor design and implementation be altered
to allow non-blocking executions and thus improve the throughput of conditionally 
synchronized multi-threaded programs on multi-core machines? 
 The motivation here is that non-blocking executions can increase parallelism by 
allowing more instructions to be executed in fixed 
time from the perspective of the thread(s). 
For this purpose, we propose  
two novel techniques:\begin{enumerate}\itemsep0em
\item A monitor object is instantiated as a thread, and invocation of monitor methods 
by application (worker) threads
is replaced by submissions of equivalent {\em tasks} to the monitor thread. The monitor 
thread assumes the responsibility of executing these tasks and returning the result 
(if required) to the application threads. 
\item The monitor methods, replaced by {\em tasks} in representation, can be executed in non-blocking manner. 
The non-blocking executions are 
explicitly indicated in the code. 
\end{enumerate}

These two changes together tackle the drawbacks discussed above. Under our proposed design, the worker threads 
do not directly access the shared data, only monitor thread does, which helps in improving the data cache locality; 
and 
non-blocking executions reduce runtime due to parallelism. 
This paper presents our experience with design and implementation of these ideas. We propose a framework, called  {\em ActiveMonitor}, 
that explores non-blocking executions of monitor methods 
 --- possibly 
at the expense of cheap and abundant processor 
and memory resources.
\input{bb_ex.tex}

Fig.~\ref{fig:bb_exp} compares the implementation of a bounded-buffer monitor written using: (a) Java's Reentrant locks
\cite{lea05}, and
(b) our proposed framework and its constructs. 
We propose two new keywords: {\sf \small activemonitor} and {\sf \small nonblocking}. 
\remove{See \cite{hg13} for details. }
The {\sffamily \small activemonitor}  keyword, in class declaration at line~$1$ of Fig.~\ref{fig:bb_exp}(b), 
 is used to declare that this bounded-buffer is implemented  
as an `active' monitor object, that is: it would exist 
 as an independent thread. 
The {\sffamily \small nonblocking} keyword in the method
signature of {\sf \small put()}, at line~$6$, indicates  
that the execution of this method can be non-blocking. 
As our experiments show (Section~\ref{sec:eval}), 
this implementation provides $\sim33$\% faster runtime 
in comparison to the conventional implementation of
Fig.~\ref{fig:bb_exp}(a). 
Lastly, our implementation is
shorter, 
cleaner, and more intuitive. 
This is because we use the {\sffamily \small waituntil}
construct (at lines $9$,$15$ in Fig.~\ref{fig:bb_exp}(b)) that provides {\em implicit/automatic signaling} mechanism 
for threads \cite{hg13}. This keyword and the automatic signaling\footnote{Appendix~\ref{app:autosynch} briefly 
covers the automatic signaling mechanism}
is not a contribution of this paper; 
we merely use it in our framework's implementation to exploit the performance benefits it provides.   

An immediate criticism of our approach  is the need for additional resources 
for its realization. Creating a new thread for a monitor object
consumes additional memory in the form of its stack allocation, as well as increases the demand 
for CPU time slices. However, it is well known 
that  
the processing power of the modern multi-core CPUs remains commonly under utilized\footnote{we collected
resource usage data, at $10$ min. intervals, of Stampede \cite{stampede} supercomputer over a period of three days. The mean  
CPU idle time was $86\%$.}. 
Additionally, with the trends of ever increasing data and computation size, most users tend to favor faster runtimes even at the cost
of somewhat increased consumption of processor and memory
resources. Our approach provides much faster runtimes
for many multi-threaded use cases. In addition, 
we claim that with careful resource management by 
the means of a 
dynamic thread creation policy that considers the processor/memory utilization, coupled with 
the automatic thread signaling mechanism \cite{hg13},
our approach's performance improvements 
come at an acceptable cost of additional CPU and memory usage. 
\remove{
In addition to the performance benefit,  
goal of our research is also to make
multi-threaded programming 
easier for the end user - while ensuring
that the underlying framework provides good runtime
performance. Using the approach of implicit thread 
notification of \cite{hg13}, we significantly 
reduce the program 
size for thread synchronization mechanism. By building
the thread-pool based monitor on top of implicit 
signaling, we are able to exploit the performance 
advantage as well explore other  
approaches such as out of order execution etc. }
In short, the key contributions of this paper are following:
\begin{itemize} \itemsep0em 
\item we propose a new monitor design, based on 
two new keywords: {\sf \small activemonitor} and {\sf 
\small
nonblocking}, that 
facilitates non-blocking executions of data updates 
within their critical 
sections. We describe the concepts required for this design's implementation in Java. 
\item we define rules to ensure correctness in presence of non-blocking operations, and show that with
our proposed design, despite additional concurrent executions, the executions 
are linearizable \cite{hw90} as per their non-blocking interpretation.   

\item we evaluate the runtimes
of some fundamental multi-threaded problems using our prototype implementation
and Java's Reentrant locks \cite{lea05}. For majority of these problems, our design and implementation provides $\sim30-40$\% faster runtimes 
in comparison to reentrant locks.
\end{itemize}

The rest of this paper is organized as follows. 
Section~\ref{sec:background} briefly covers the background 
concepts. Section~\ref{sec:concept} 
discusses the key concepts  of the framework
and correctness of executions. Section~\ref{sec:eval} presents experimental evaluation of our approach.
We discuss some limitations of our approach, and future work in Section~\ref{sec:discussion}. Section~\ref{sec:related} presents 
the related work, and Section~\ref{sec:conclusion}
offers concluding remarks.

%% file: bb_ex.tex
\begin{figure*}[ht!]
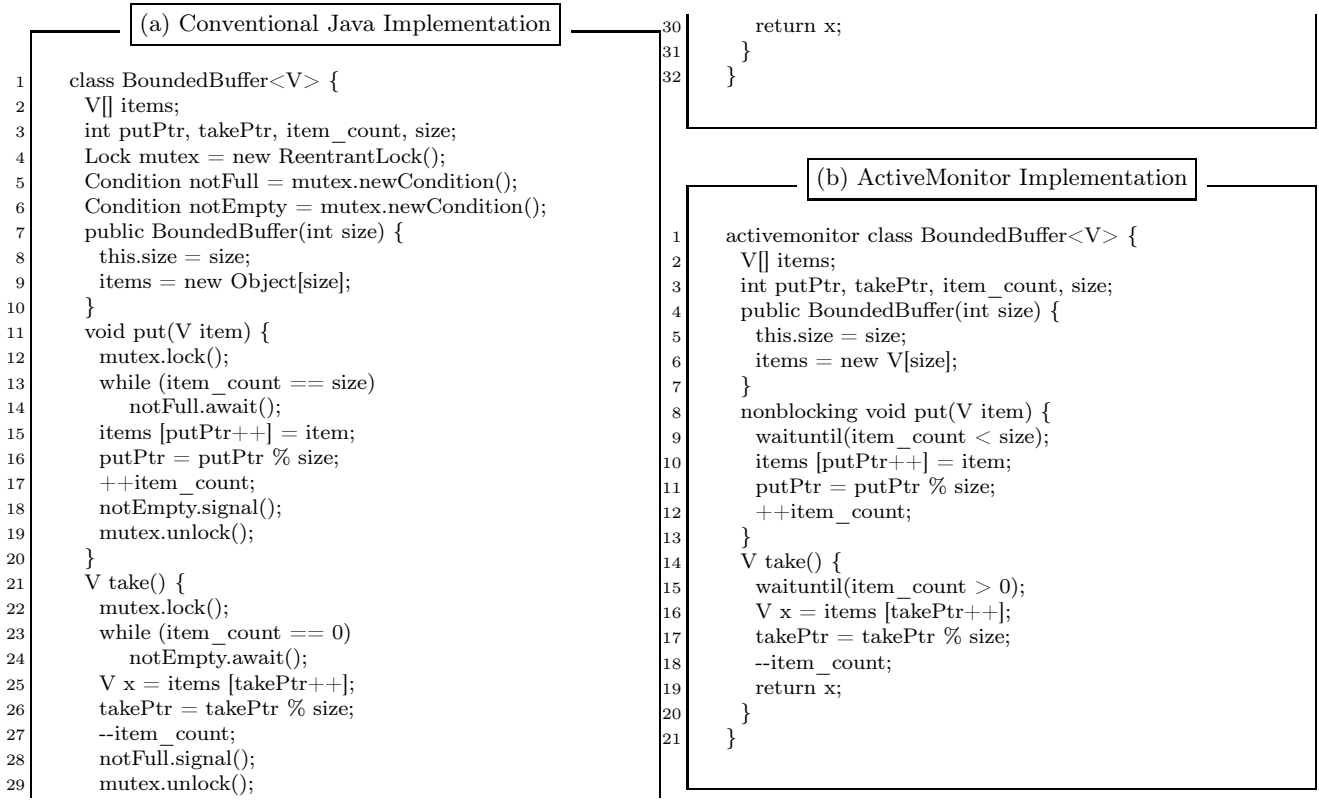

\begin{multicols}{2}
    \begin{Verbatim}[fontfamily=serif,fontsize=\scriptsize,gobble=8,frame=single,
            framesep=5mm,numbers=left,numbersep=2pt,
            label=\fbox{\footnotesize \emph{(a) Conventional Java Implementation}}]
        class BoundedBuffer<V> {
          V[] items;  
          int putPtr, takePtr, item_count, size;
          Lock mutex = new ReentrantLock();
          Condition notFull = mutex.newCondition();
          Condition notEmpty = mutex.newCondition();
          public BoundedBuffer(int size) {
            this.size = size;
            items = new Object[size];
          }
          void put(V item) {
            mutex.lock();
            while (item_count == size)
                notFull.await();
            items [putPtr++] = item;
            putPtr = putPtr % size;
            ++item_count;
            notEmpty.signal();
            mutex.unlock();
          }
          V take() {
            mutex.lock();
            while (item_count == 0) 
                notEmpty.await();
            V x = items [takePtr++];
            takePtr = takePtr % size;
            --item_count;
            notFull.signal();
            mutex.unlock();
            return x;
          }
        }
    \end{Verbatim}
    \begin{Verbatim}[fontfamily=serif,fontsize=\scriptsize,gobble=8,frame=single,framesep=5mm,
            numbers=left,numbersep=2pt,
            label=\fbox{\footnotesize \emph{(b) ActiveMonitor Implementation}}]
        activemonitor class BoundedBuffer<V> { 
          V[] items; 
          int putPtr, takePtr, item_count, size;
          public BoundedBuffer(int size) {
            this.size = size;
            items = new V[size];
          }
          nonblocking void put(V item) { 
            waituntil(item_count < size); 
            items [putPtr++] = item; 
            putPtr = putPtr % size;
            ++item_count; 
          } 
          V take() { 
            waituntil(item_count > 0); 
            V x = items [takePtr++]; 
            takePtr = takePtr % size;
            --item_count;
            return x;
          }
        }
    \end{Verbatim}
\end{multicols}
  \caption{Comparsion of Bounded-Buffer programs written using (a) Java's locks and (b) our approach}
  \label{fig:bb_exp}
\end{figure*}

%% file: background.tex
\section{Background}
\label{sec:background}
We use the concept of {\em futures} \cite{future} to achieve non-blocking executions of critical sections.  \\
\hfill\\
{\bf Futures:} 
Futures \cite{future} were proposed as a 
mechanism to facilitate task submissions (to threads) in shared memory parallel programs. 
In a task submission based model some threads, called {\em executors}, provide task execution as a service. Under this model, a thread makes an asynchronous
call to some executor thread, and a {\em future} object is returned to the caller thread as a pointer to the computation and its possible result. 
The executor thread performs a task for the caller thread, and updates the returned future object with the result of the 
computation of the task. 
\remove{A simple example 
of usage of a {\em future} object is the following:
\begin{Verbatim}[fontfamily=serif,frame=single]
Future<V> f = invoke_method(args); 
// execute some other instructions
...
//evaluate the future's result
V result = f.get();
\end{Verbatim}
}
The caller thread can collect the result  
 by {\em evaluating} the future. Evaluation of a future,
i.e. the collection of the result, is performed by invoking a
blocking method on the future object. Hence, if needed the thread would have 
to wait for the task computation to complete. In our framework, we use futures to realize non-blocking 
executions of monitor methods by treating method invocations as task submissions.\\
\hfill\\
{\bf Blocking and Non-Blocking methods}: 
By design, methods provided by monitor objects are usually blocking ({\em synchronous}). If a thread invokes a blocking 
method $m$, then the thread can not execute 
its next instruction until it has exited
$m$; whether by normally returning from the method
or exiting the method due to an exception. Whereas,    
if a thread invokes a 
non-blocking method $m'$ then it would be able to execute its next instruction
immediately after, while the non-blocking method is being 
executed in parallel, possibly by another thread.

%% file: concept.tex
\section{{\em ActiveMonitor} Concepts} \label{sec:concept}

In this section, we discuss the key concepts of the {\em ActiveMonitor} framework. The relevant implementation details are presented in Appendix~\ref{app:impl}. 
We first explain our approach using the bounded-buffer monitor of Fig.~\ref{fig:bb_exp}(b). Observe that 
apart from using the implicit-signal construct of {\sf \small waituntil}, the code uses two additional
keywords that we propose: {\sf \small activemonitor} and {\sf \small nonblocking}. 
By including the    
\actm keyword in the declaration (eg. line $1$), the programmer indicates that this monitor is an `active' object, i.e. 
it should be instantiated as a thread. However, if there are many monitor objects in a user's program  
we limit the number of 
monitor threads --- this is an implementation detail, and discussed in Appendix~\ref{sec:threadpool}. In the conventional 
monitor design, worker threads, i.e. threads explicitly created by the user program, execute the instructions of the critical section and update the data protected by the
monitor.  In contrast, as per our design, 
once the monitor is instantiated as a thread, it executes its own critical sections on behalf of worker threads. 
For this, at compile time method calls made to the monitor are treated as task submissions to the monitor thread.
\remove{In our approach,
only the monitor thread executes its own critical sections. Thus, the monitor 
thread 
 assumes the responsibility of finishing these tasks for the threads that submitted them.} Thus, 
 for the program in Fig.~\ref{fig:bb_exp}b, 
 a producer thread's call 
to the {\sf \small put} method, as well as a consumer thread's call to the {\sf \small
take} method, are converted to submission of equivalent tasks to the monitor thread by our framework. The worker 
threads now have the option to perform blocking ({\em synchronous}) or non-blocking ({\em asynchronous}) executions of monitor methods. 
Asynchronous executions are made possible by using our second keyword proposal, {\sf \small nonblocking}, 
that allows for non-blocking (asynchronous) 
execution of instructions of the method using this keyword. We first describe how monitor tasks are generated from monitor methods, and then 
discuss their asynchronous or synchronous executions. 
\remove{For ease 
of explanation, let us assume that there is just one monitor in the system. 
uses two keywords: {\sf \small activemonitor} and {\sf \small nonblocking}, in her 
monitor program. , that the monitor object should be instantiated as a thread. We call these threads monitor threads. In order to limit the overhead of the monitor threads, we use
a dynamic approach} 

\subsection{Monitor Tasks}

In our framework, every monitor method call corresponds to an equivalent task. At compile-time, invocations of monitor methods by worker 
threads are replaced by submissions of the equivalent tasks to the monitor thread.  
\begin{definition}[Monitor Task] A monitor task $t$ consists of a 
    boolean predicate $P$ and a set of statements $\mathcal{S}$. If the precondition defined by $P$ is true then $t$ is 
    `executable' and
    statements in $\mathcal{S}$ 
    can be executed to complete $t$.  Otherwise, $t$ is `unexecutable'. 
   \end{definition}
\remove{   
\begin{figure}[ht!]
\begin{Verbatim}[fontfamily=serif,frame=single,numbers=left,xleftmargin=5mm]
nonblocking void insert(V item) {
   waituntil (itemCount < buffer_size);
   items [putPtr++] = item;
   putPtr = putPtr % size; 
   ++item_count;
}
\end{Verbatim}
    \caption{Example of a non-blocking task}
    \label{fig:task_ex}
\end{figure}
}

\remove{While generating a monitor task from a monitor method, 
our framework guarantees that the set of statements of the derived task 
includes every statement in the original monitor method. }
Observe that 
the syntax of a monitor task is similar to that 
of a guarded command \cite{dij75}. Consider the {\sf \small put} method 
(lines $8-13$) of the bounded-buffer in Fig.~\ref{fig:bb_exp}b.
For this monitor method, the equivalent monitor task $T$ is defined by the code 
of lines $9-12$. Here, the precondition $P$ is 
{\sf \small (itemCount < buffer\_size)}; and it checks 
if the buffer has any space to insert the item. If this condition is false, the {\sf \small waituntil}
construct ensures that any thread trying to complete this task has to wait unless 
the buffer is not full. Lines $10$ and $11$ together 
form the set of statements $\mathcal{S}$. As indicated by the {\sf \small nonblocking} keyword in the method signature
at line $8$, the 
generated task would be submitted for a non-blocking 
execution to the monitor thread. Removal of this keyword from the method signature makes the execution blocking. 
\remove{
\begin{figure}[ht!]
\begin{multicols}{2}
    \begin{Verbatim}[fontsize=\footnotesize,gobble=6,
            codes={\catcode`$=3\catcode`_=8}]
        void foo() {
            $waituntil(P_1)$;
            $S_1$;
            $waituntil(P_2)$;
            $S_2$;
            ...
            $waituntil(P_n)$;
            $S_n$;
        }
    \end{Verbatim}
    \columnbreak 
    \begin{Verbatim}[fontsize=\footnotesize,gobble=2,
            codes={\catcode`$=3\catcode`_=8}]

        void bar() {
            $S_1$;
            $waituntil(P)$;
            $S_2$;
        }
    \end{Verbatim}
\end{multicols}
    \caption{Examples of two monitor methods.}
    \label{fig:method_ex}
\end{figure}

Intuitively, a $waituntil(P)$ statement corresponds to a precondition $P$
of a task $T$; the statements below the $waituntil(P)$ and above another
$waituntil$ statement constitute the set of statements $S$ of $T$. 
For example, in Fig.~\ref{fig:method_ex}, the first derived task $T_1$ from
the $foo$ method has the precondition $P=P_1$ and the set of statements $S_1$.
}

 The details of compile time generation of tasks for monitor method calls are in Appendix~\ref{sec:app_preproc}. 
 Appendix~\ref{app:task_ex} discusses some special cases of monitor tasks. 
 Our current prototype cannot convert 
some particular types of monitor methods to equivalent tasks -- for example, methods with {\sf \small waituntil} 
not at the beginning of the critical section, or {\sf \small waituntil} 
in a conditional branch. They are also discussed in Appendix~\ref{sec:app_preproc}. 


\subsection{Monitor Tasks: Execution and Correctness Guarantees}
\label{sec:threadpool}
In the {\em ActiveMonitor} framework, monitor threads are responsible for executing monitor tasks. 
Based on the availability of resources, the framework tries to create one thread per monitor object.
The implementation details for limiting the monitor thread count are in Appendix~\ref{sec:threadpool}. 
Monitor threads execute monitor tasks while   
observing the following rules. 
\begin{mrule}[Mutex Invariant] 
All the tasks of a single monitor object are executed by the same monitor thread. 
\remove{    monitor $M$. If a thread is executing 
    task $m \in M$, then no other can be executing
    $m$ at the same time. }
    \label{def:mutex}
\end{mrule}
\remove{
The above rule only prohibits concurrent executions of tasks 
and not repeated executions. This is because it is possible that a programmer may want re-execution of tasks
whenever exceptions occur. Details of our framework's exception handling features are presented in Appendix~\ref{app:exception}. 
allows customized exception handling, and thus it is possible that a task 
is executed more than once. Exception handling for asynchronous executions is performed by  }
Thus, Rule~\ref{def:mutex}
maintains the mutual exclusion of critical sections of a monitor. However, for 
correctness of executions, we require additional rules. 
Let $proc(t)$ denote the thread that submits the task $t$ to the
a monitor; $sub(t)$ be the time at which the task $t$ is submitted to the monitor, and $exe(t)$ indicate the time when
 the monitor thread starts
executing $t$. 
\begin{mrule}
\label{rule:sc}
For a pair of tasks $s$ and $t$
    submitted to a monitor $M$, if $proc(s) = proc(t)$, then\\
       $ sub(s) < sub(t) \Rightarrow exe(s) < exe(t)$.
\end{mrule}
\begin{mrule}
\label{rule:linear_active1}
Let $n_1$,~$m_2$ be two successive method invocations by a worker thread on two different monitors $M_1$ 
and $M_2$ in the user program, and let $t_1$,~$t_2$ be the corresponding task submissions at runtime. Then, 
at runtime $t_1$ must be completed before  
$t_2$'s submission. 
\end{mrule}
The notions of method {\em invocation} and {\em response} used to define linearizability \cite{hw90} need a different interpretation under non-blocking executions. In short, {\em invocation} now corresponds 
to submission of the equivalent task to monitor thread, and {\em response} corresponds to 
this tasks completion. With this interpretation, the following 
result establishes correctness of executions by proving linearizable executions. 

\begin{lemma}
\label{lem:lock} With
rules~\ref{def:mutex},~\ref{rule:sc}, and~\ref{rule:linear_active1} we get an execution that 
is equivalent to a lock-based execution. 
\remove{together 
 guarantee linearizability with method invocations being interpreted as task submissions, and 
 method response interpreted as task completion. }
\end{lemma}
\begin{proof} Sketched in Appendix~\ref{app:proof_lock}. \end{proof}
\remove{
\begin{proof}
\remove{
Sketch: We show linearizablity by showing that enforcing Rules~\ref{def:mutex},~\ref{rule:sc}, and~\ref{rule:linear_active1} 
results in an execution that is equivalent to a lock-based execution. Since lock-based executions are linearizable, 
the result follows. 

Proof:}
We show that for any execution in our model there exists 
an equivalent lock-based execution.
Since all tasks of any monitor object are executed by a single thread 
due to Rule~\ref{def:mutex}, mutual exclusion is preserved just as in any lock-based
execution. We only need to show that the order of execution of the tasks
corresponds to a schedule in which worker threads execute the tasks.

\remove{we show that the sequence of tasks executed corresponds to a
sequence possible for a lock-based execution.}
It is sufficient to show that all tasks submitted by a single worker
thread execute in the order of submissions.
Let $s$ and $t$ be two consecutive tasks submitted by the thread.
If they are submitted for the same monitor, then the Rule~\ref{rule:sc} preserves the order. 
If $s$ is a blocking task, then by definition of blocking task, 
$t$ cannot submitted before $s$ 
is completed. Hence, execution of $s$ precedes execution of $t$.
If $s$ is a non-blocking task and is on a different monitor object from $t$,
then due to Rule~\ref{rule:linear_active1} we wait for $s$ to finish before submission of $t$.
\end{proof}
}
Observe that Rule~\ref{rule:linear_active1} forces a worker thread to wait for the previous 
task to finish even if that task was on a different monitor object. 
In many applications, a programmer may not
require this constraint on different monitor objects. 
Hence, to improve performance, we can drop 
Rule~\ref{rule:linear_active1} for such applications. We show that the execution is still linearizable with 
the new interpretation of invocation and response. 

Note that in the standard model of concurrent history \cite{hw90},
the thread history is always sequential although the object
history may not be sequential. In our model, due to non-blocking executions, we have the dual
property: an object history is always sequential whereas a thread
history may not be so. We first define the thread order
in presence of non-blocking operations. Let 
$s_1, s_2,...,s_m$ be $m$ operations performed on a monitor object by a worker thread. Let $s_i < s_j$ denote that 
$s_i$ was executed before $s_j$. 
In the standard model, all operations are blocking and 
we get that $s_i < s_{i+1}$ for all $i$. 
In our model, we define the order as follows.
If $s_i$ is blocking, then $s_i<s_j$ for all $j > i$.
If $s_i$ and $s_j$ are operations on the same object and $i < j$, then $s_i < s_j$.
If $s_i$ is non-blocking and the result of $s_i$ is required
before $s_k$, then $s_i < s_j$ for all $j \geq k$.

\begin{lemma} \label{lem:linearizable}
With Rules~\ref{def:mutex} and \ref{rule:sc}, we get an execution that
is linearizable.
\end{lemma}
\begin{proof} Sketched in Appendix~\ref{app:proof_linearizable}. \end{proof}
\remove{
\begin{proof}
It is sufficient to show that
for every thread history there exists an equivalent
sequential thread history that is consistent with execution by the monitor thread.
We get this sequence by considering as linearization point
for an operation the instant at which the monitor thread finishes
executing the corresponding task. We show that this order
is consistent with the thread order.

First consider the case when $s < t$ because $s$ is a blocking operation.
In this case, task $t$ cannot be submitted before task $s$
is executed. Hence, the order of execution of monitor tasks preserves the thread
order. Second, consider the case
when $s < t$ because they are operations on the same monitor object and $s$ is
invoked before $t$. In this case Rule~\ref{rule:sc} guarantees that the order of execution
is $s$ followed by $t$.
Finally, consider the case when $s$ and $t$ are tasks submitted by 
the same thread, $s$ is non-blocking but its results are used by $t$.
In this case, our implementation blocks for the results rom $s$ before
executing any further. Hence $t$ would be submitted later than the execution 
of $s$.
\end{proof}
}

Observe that the legal sequential history we get may not preserve the order
of invocation of operations, but only the thread order.

%% file: evaluation.tex
\section{Evaluation} \label{sec:eval}
We now present the experimental evaluation of our proposed
approach. We evaluate our prototype implementation  
on five multi-threaded problems. 
The first problem is a bounded-buffer problem with multiple producers and consumers, the second is  
a sorted linked-list (of integers) on which threads perform updates. 
The remaining three are 
conditional 
synchronization problems involving varying levels of complexity in their conditional predicates. 
 We implemented the  
 problems 
 using the {\em ActiveMonitor} framework, Java's reentrant locks, and AutoSynch \cite{hg13}. Additional 
 use-case specific
 implementations (details given in problem descriptions below) are also compared for bounded-buffer, and the linked-list
 problem. 
For all the problems,
we use a benchmark in which each worker thread
performs $512000$ operations on shared
data protected by monitors. 
All the experiments are conducted on a machine with
$1$ socket, $2$ AMD Opteron $6180$
SE $12$ Core, $2.5$ GHz CPUs (= $24$ hardware threads) and $64$ GBs memory running Linux $2.6.18$.
We vary the number of worker threads, and measure the time 
required for all of them to complete their operations. All threads perform a fixed number 
of warm-up operations on shared data before starting the time measurements. 
For each problem, we measure runtimes for
$25$ runs, and report the mean values 
 after removing the highest and the lowest values.    

For all the experiments, we restrict the {\em ActiveMonitor} framework  to create only one monitor thread. 


We now discuss the setup of experiments and their results in a problem specific manner. 
Across all results, we denote the implementation techniques with the following
notation: {\sf \small LK}: Reentrant locks, {\sf \small AS}: {\em
AutoSynch}\cite{hg13}, {\sf \small AM}: {\em ActiveMonitor} (this paper).  For
each problem, all the runtimes reported in their results exhibit negligible variance ($<5\%$ of the mean
value) across runs.  

\subsection{Bounded-Buffer (BB) problem \cite{dij65, dij71}}

{\em \underline{Setup}}: Items are plain Java objects. Every producer's 
        {\sf \small put} invocation is 
        non-blocking, and 
        every consumer's {\sf \small take} is blocking. We also compare
runtimes of Java's {\sf \small ArrayBlockingQueue} \cite{lea05}, denoted by {\sf \small ABQ},  
 based implementation. Number of producers (= number of consumers) is varied from $2$ to $24$.  
We perform three types of experiments, and measure: (a) runtimes wieh varying number of threads 
for a buffer of fixed size (=$4$). (b) runtimes with varying buffer-size for fixed number of producers/consumers (=$16$ each). (c) runtimes with varying limit on non-blocking tasks allowed on the buffer of size $4$ with $16$ producers and consumers each.   
\begin{figure*}[t!]
\centering
\begin{subfigure}[b]{0.32\textwidth}
  \includegraphics[height=42mm]{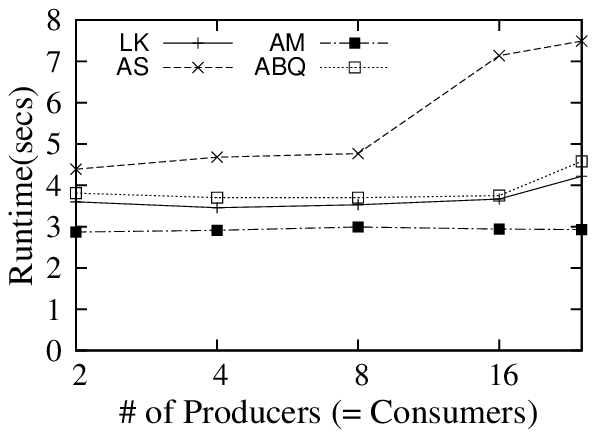}
  \caption{Buffer size~$=4$}
  \label{fig:pc_eval}
\end{subfigure}
\begin{subfigure}[b]{0.33\textwidth}
  \includegraphics[height=42mm]{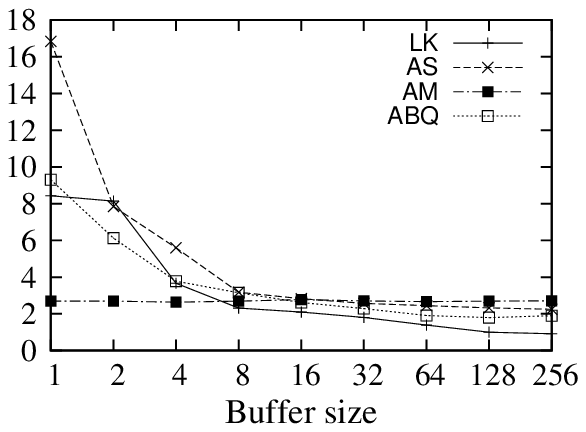}
  \caption{Varying buffer size}
  \label{fig:vpc_eval}
\end{subfigure}
\begin{subfigure}[b]{0.32\textwidth}
  \includegraphics[height=42mm]{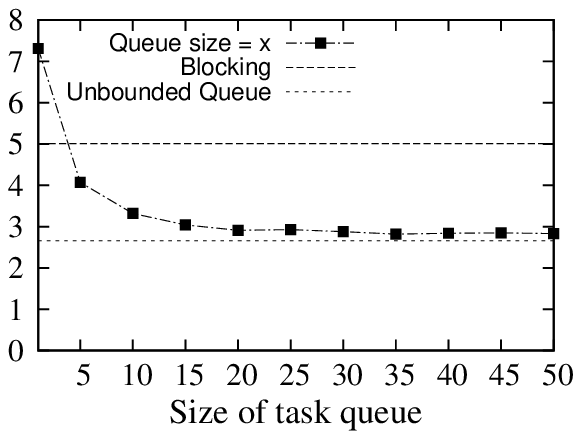}
  \caption{Varying queue size for {\em tasks}}
  \label{fig:vqpc_eval}
\end{subfigure}
    \caption{Runtimes for Bounded-Buffer problem}
  \label{fig:bbruntimes}
\end{figure*}

\begin{flushleft} {\bf Results}\end{flushleft}
  Collective results are shown in Fig.~\ref{fig:bbruntimes}.  Fig.~\ref{fig:pc_eval} plots the runtimes
for experiment (a). Fig.~\ref{fig:vpc_eval} shows the runtimes for experiment (b). This result
highlights the benefits of non-blocking executions. Recall that 
only {\sf \small put} calls are non-blocking for {\sf \small AM}; all
{\sf \small take} calls are blocking. For small size buffers, \lk, \as, and {\sf \small ABQ} are much slower in comparison
to \am ~due to their blocking insert
operations.  For larger buffer sizes, \lk~ and {\sf \small ABQ}
implementations perform slightly faster than {\sf \small AM}. On larger buffer sizes, the frequency
of threads getting blocked out is 
much lower, and both \lk~  and \abq~ 
benefit from greedy implementation of Java in lock acquisition; 
but \am~ends up performing slightly more work in enforcing order on tasks (as per Rules~$1$,$2$). 

 Fig.~\ref{fig:vqpc_eval} shows the results for experiment (c). `Limit' on the size of task queue 
 means that whenever the queue maintained by the monitor thread to store tasks, reaches the `limit' value 
 then even a non-blocking task submission to the monitor is force to block. 
Making all the invocations blocking (queue size limit = $0$) leads to an overall runtime (for
completion) of $\sim5$ seconds, whereas with no limit on the
size of task queue the runtime is
 $\sim3$ seconds. This observations clearly shows that our non-blocking task based approach is beneficial. 
However, an unbounded task queue might grow quite large. But as shown by the plot, 
 even for a short limit (around $20$) on the size of this queue, the runtime performance converges to optimal
performance (of unbounded size).  Thus, this result establishes that 
only a small number of non-blocking tasks can provide significant
performance benefit.

\subsection{Sorted Linked-List (SLL) Problem}
{\em \underline{Setup}}: 
A linked-list of integers sorted 
in non-decreasing order. The number of worker threads is varied from $2$ to $24$. 
The list is pre-populated
with $5000$ integers. Each worker thread
inserts or removes (with equal probability) a random integer. For a fair comparison, seeds are used for randomization 
so that the threads generate the same sequence of random integers across runs as well as across 
different monitor implementations.  We make every thread execute some local instructions, simple
mathematical computations, outside the critical section between any successive calls on the monitor.  
This is done in attempt to simulate some practical program behavior in which generation of data
takes some time\footnote{
Moreover, 
The results for the BB problem already showed that we perform better than conventional monitors for saturation tests.}. 
 Both insert and remove operations are 
non-blocking. Additional techniques used for performance comparison: lock-free implementation ({\sf \small LF}) \cite{hs08, har01, mic02}, 
fine-grained
locking ({\sf \small FG})  \cite{hs08}, and transactional memory ({\sf \small TM})
implementation in Java \cite{ksf10}. We measured runtimes for two types of experiments: 
(a) runtimes with varying number of threads, when each thread performed $250$ mathematical computations (integer additions)
outside the critical section between successive calls to the monitor. (b) runtimes with varying 
number of threads and varying number of operations performed outside the critical section. 
\begin{figure}[t!]
\centering
\begin{subfigure}[b]{0.40\textwidth}
  \centering
  \includegraphics[height=42mm]{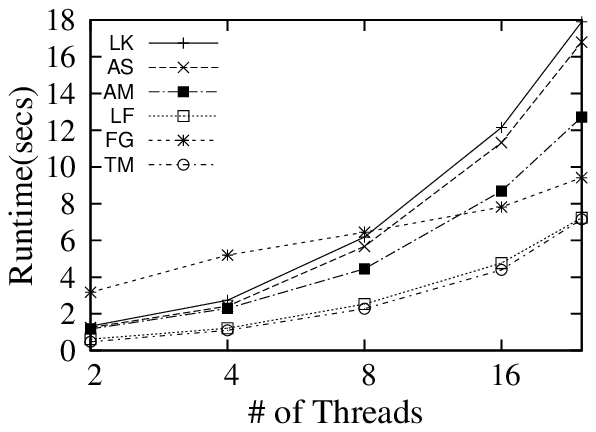}
  \caption{Runtimes for experiment (a)}
  \label{fig:list_eval}
\end{subfigure}
\begin{subfigure}[b]{0.59\textwidth}
  \centering
  \includegraphics[height=42mm]{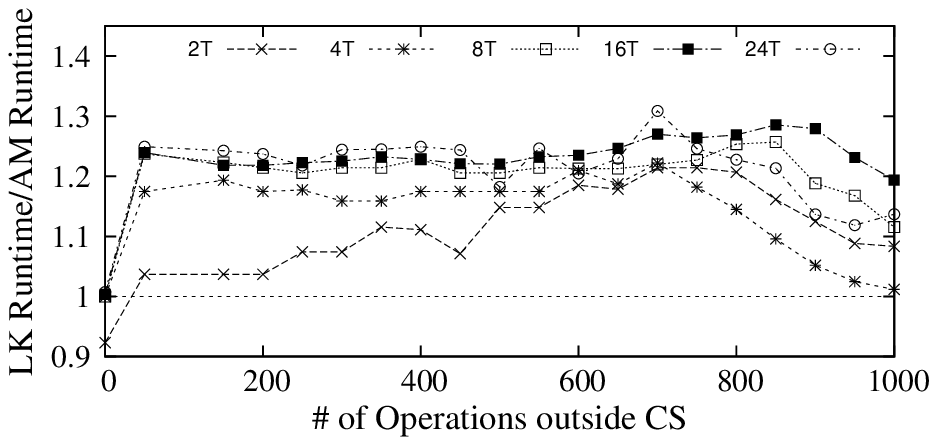}
  \caption{Runtime ratio with varying number of operations outside CS}
  \label{fig:vsin_eval}
\end{subfigure}
  \caption{Results for sorted linked-list problem}
  \label{fig:llruntimes}
\end{figure}
\begin{flushleft} {\bf Results}\end{flushleft}
Fig.~\ref{fig:list_eval} shows the results for experiment (a). For this experiment, the {\em ActiveMonitor} implementation ({\sf \small AM})
clearly outperforms the reentrant lock based monitor ({\sf \small LK}), and
{\em AutoSynch} ({\sf \small AS}). However, lock-free ({\sf \small LF}) and transactional memory ({\sf \small TM})
based implementations perform better. This is because both of them use
optimized algorithm for the linked-list data structure. Whereas, we do not use
any linked-list specific optimization, and only use non-blocking executions on
critical sections. 

Fig.~\ref{fig:vsin_eval} plots the ratio of runtimes using \lk~
implementation and \am~implementations for experiment (b). The ratio with
no operations outside the critical section is close to $1$, i.e. our
approach provides no performance benefit. This is expected because no
operations can be performed concurrently with the non-blocking updates to the
list, and the
opportunity of parallelism is effectively absent. As the number of operations outside the
critical section is increased, we see significant performance gains with our approach, upto $\sim30\%$, due to
increased parallelism introduced by non-blocking executions.  

\subsection{Conditional Synchronization Problems}
We used three thread synchronization problems that involve different levels of complexity in their
conditional predicates. For all three, we measured the time taken for all the worker threads to complete the
designated number ($512000$)  
of operations on the monitor.   \\
\hfill\\
{\bf Round-Robin (RR) Access}: Every worker thread accesses the
        monitor in round-robin order. All operations on the monitor are non-blocking.
        The number of worker threads is varied from $2$ to $24$.\\
\hfill\\
{\bf Parametrized Bounded-Buffer (PBB) \cite{dij65, dij71}}: Producers put a collection of items 
into a shared buffer, while consumers remove a number of items from the buffer. For producers, the number of items to 
be inserted, 
and for consumers, the count of items to be removed, both are randomly decided at runtime. Similar to the SLL problem, seeds 
are used for randomization for fair evaluation across implementations.   The number of producers (= number of 
consumers) is varied from $2$ to $24$.\\
\hfill\\
{\bf Ticketed-Readers/Writers (TRR) \cite{chp71,bh05b,han72}}:
    A
    ticket is used
        to maintain the access order for readers and writers \cite{bh05b}. Every reader/
         writer gets a ticket number indicating its arrival order. Readers/writers wait on the monitor for their turn. 
        On their turn, they enter the monitor but  
        do not perform any computation inside
        the monitor, and immediately exit. Operations for arrival on monitor and waiting for turn to access 
        the monitor are blocking, 
        whereas operations after getting access to critical sections are non-blocking. The number of writer threads is varied from $2$ to $24$. The number of readers is  
         kept $5$ x the number of writers. 

\begin{figure*}[t!]
  \centering
\begin{subfigure}[b]{0.33\textwidth}
  \centering
  \includegraphics[height=42mm]{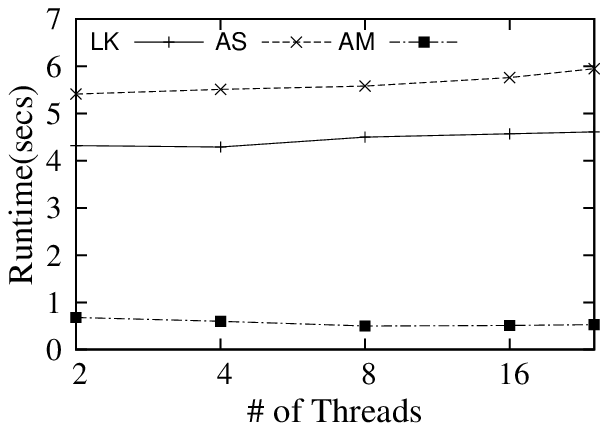}
  \caption{Round-robin problem}
  \label{fig:rr_eval}
\end{subfigure}
\begin{subfigure}[b]{0.32\textwidth}
  \centering
  \includegraphics[height=42mm]{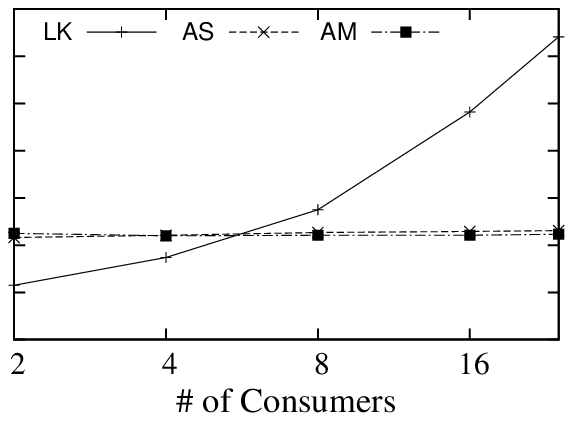}
  \caption{Parametrized bounded-buffer}
  \label{fig:ppc_eval}
\end{subfigure}
\begin{subfigure}[b]{0.32\textwidth}
  \centering
  \includegraphics[height=42mm]{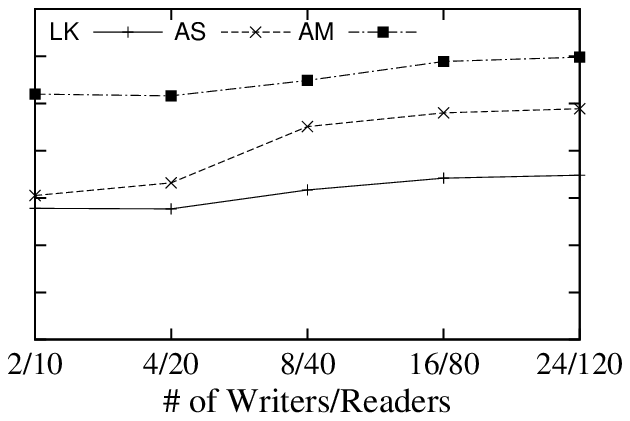}
  \caption{Ticketed r/w problem}
  \label{fig:trw_eval}
\end{subfigure}
\caption{Runtimes for RR, PBB, and TR/W problems: all plots have the same range on y-axis}
  \label{fig:rrruntimes}
\end{figure*}

\begin{flushleft} {\bf Results}\end{flushleft}
Fig.~\ref{fig:rrruntimes} presents the runtimes of {\sf \small LK}, {\sf \small AS}, and {\sf \small AM} based implementations
for the three problems. 
For the RR problem, the performance of {\sf \small AM} is $~6$ and $~5$ times faster than that of {\sf \small AS}, and
{\sf
\small LK}, respectively. 
This is due to better cache 
locality and fewer context switches. For this problem, after accessing the critical 
section, every thread is forced to wait unless all other threads have accessed the 
critical section. This constraint leads to poor cache locality and more context switches for 
conventional monitor design. 
However, in our approach, the monitor thread continuously executes all monitor
accesses for each thread; therefore, there is minimal number of context switches between each 
monitor method invocation and that results in excellent cache locality.

For the TR/W problems, {\sf \small AM} performs poorly in comparison to {\sf \small LK}. This is primarily due to the nature of the problem and its underlying implementation 
requirements. The problem requires conditional synchronization in the middle of the method calls, 
and thus our task based approach has to create two tasks per method call. In fact, any family of problems with 
such characteristics would lead to much higher number of tasks than the actual method calls. 

\subsection{CPU and Memory Consumption}
Table~\ref{tab:cpu_mem} presents the results on the CPU and memory usage by {\sf \small LK}, {\sf \small AS}, and {\sf \small AM}
for 
all the evaluated problems. For each problem, the CPU ratio stands for the ratio of an implementation's mean (across runs)
CPU consumption   
and mean (across runs) CPU consumption of {\sf \small LK} implementation. As per memory usage, we report the maximum 
memory (in MBs) consumed across any of the runs performed by the implementations for each problem. As expected, our approach consumes more CPU resources. However, the maximum memory consumed by our approach is only incrementally 
larger than that of Java's reentrant lock based approach. Thus, we have backed our claim that with careful resource management
the cost for the performance benefit is relatively small. 
\begin{table*}
    \centering
    \footnotesize
\begin{tabular}{|l||r|r|r||r|r|r||r|r|r||r|r|r||r|r|r|}
\hline
Problem & \multicolumn{3}{c||}{{\bf BB}} 
          & \multicolumn{3}{c||}{{\bf SLL}}          &   
\multicolumn{3}{|c||}{{\bf RR}} &
\multicolumn{3}{c||}{{\bf PBB}} & 
          \multicolumn{3}{c|}{{\bf TR/W}}              \\ \hline 
Approach          & {\sf \small LK}& {\sf \small AS}& {\sf \small AM}& {\sf \small LK}& {\sf \small AS}& {\sf \small AM}& 
{\sf \small LK}& {\sf \small AS}& {\sf \small AM}& {\sf \small LK}& {\sf \small AS}& {\sf \small AM}& {\sf \small LK}& {\sf \small AS}& {\sf \small AM}\\ \hline
CPU Ratio & 1.00 &  1.11 & 1.95 
                    & 1.00 & 0.86 & 3.49 & 
1.00 & 1.10 & 2.97 
& 1.00 &  2.39 & 2.47
& 1.00 & 1.23 & 1.24 \\ 
Max Mem.    & 196 &  236 & 197 
                    & 135 & 135 & 145             & 
236 & 224 & 250 
& 196 & 227 & 251 
& 211 & 232 & 244 \\ \hline
\end{tabular}
\caption{CPU and memory consumption with Reentrant locks ({\sf \small LK}), {\em AutoSynch} ({\sf \small AS}), and {\em ActiveMonitor} ({\sf \small AM}).} 
\label{tab:cpu_mem}
\end{table*}





%% file: discussion.tex
\section{Discussion \& Future Work}
\label{sec:discussion}
The benefits of our approach are manifold: improved runtime performance, better control over desired program behavior, 
and increased opportunities for applying optimizations such as execution re-ordering, operation combining, etc. However, 
some aspects of our framework are open to criticism.  
 First and foremost, based on the design assumptions as well as the experimental results, it is evident that at present, our approach 
does not suit systems on which processor and memory resources are scarce. Even though our evaluation results indicate 
slightly increased memory consumption, the increase in CPU consumption is more drastic. 
Hence, an important future work is an improved implementation of our framework that lowers the CPU usage. 
Another future work is to 
address the issue of task generation for problems that require conditional synchronization in the middle 
of the method call, eg. Ticket-R/W, and causes large number of tasks to be created. 
It is important 
to note that our approach extends conventional multi-threaded programs and their constructs, and thus provides a programmer
the choice of an implementation she thinks is most suited for her needs. 

We showed that we provide linearlizability of executions --- interpreted under the non-blocking execution model. Although 
it is desirable in most of the cases, there are 
situations when this requirements can be relaxed 
to certain extents
for even faster performance. Such relaxations are 
very much dependent on the use-cases and the 
expectations of the programmer.  
As shown by Kogan et al. in  
\cite{kogan_future}, having `weak' consistency requirements could lead to 
significant performance gains. In fact, for specific use-cases programmers 
may only need {\em eventual consistency} \cite{vog09}, which requires that 
any update to shared data is eventually committed, in their 
programs. However, allowing 
weak or eventual consistency is not possible for all use-cases. For example, consider multi-threaded 
implementations (in shared memory model)
of Dijkstra's single-source-shortest-path algorithm \cite{sssp, delta}. Weak or eventual consistency based implementations for this algorithm 
cannot guarantee correct results. However, for a large number of problems, approximate solutions are good enough, and their parallel 
implementations exhibit significant gains in performance. Thus, this opens up two interesting future problems for our work.
First of them is  
 evaluating performance of our non-blocking execution approach for problems that admit approximate solutions, and 
 the second is to extend our evaluation to the weak consistency model of \cite{kogan_future}.  

When it comes to evaluating the performance of our proposed design and its implementation, the exploration space is large. 
Not only there is a large number of problems related to multi-threaded updates on data structures, but the 
set of parameters that affect the 
runtime for these problems is also large. 
We also plan to perform experiments by varying more parameters such as load distributions of methods, size of shared data, etc.
Another key experiment is to analyze the cache locality of our approach --- in terms 
of cache miss rates per operation and per instruction --- for different experiments, 
 and comparing it to cache locality achieved by other techniques. 

Lastly, note that our current implementation does not use ideas such as
operation combining \cite{his+10}, 
 execution re-ordering, and batch updates \cite{kogan_future}.  Intuitively, incorporating 
these strategies in our implementation could lead to even better results than those  
observed in this paper. 

%% file: related.tex
\section{Related Work}
\label{sec:related}
The idea of having
monitor objects execute as independent
threads is not completely novel. Hoare proposed 
a similar
mechanism, in which all objects are {\em active},   
of communicating sequential processes (CSP) in \cite{csp} long ago. However, the proposal 
was a formal language for interactions between processes. Its constructs
are still mostly used for formal analysis. Our work focuses on making parallel programs faster
by using {\em futures} 
to provide selective non-blocking 
executions on monitor methods. 

Research efforts to increase parallelism in programs 
have been continuous almost since the advent of computing.  With the arrival of multi-core processors, 
improving multi-threaded programming with syntactic 
changes, as well as with novel designs
of data structures and algorithms, 
has been the focus of many research efforts in the last 
decade. Transactional
memory \cite{hm93, st95} is 
a well-known research effort that proposes
modified syntax for ease of writing multi-threaded programs. The prominent contributions in algorithmic direction are:
lock-free data structures 
\cite{har01,mic02,hsy04}, and wait-free
data structures \cite{kp11, kp12, tbk+12}. Some other research efforts have explored increased parallelism 
for data structures by operation re-ordering and combining complementary operations (insert and remove) \cite{hhl+06}. 

{\bf Transactional memory} \cite{hm93, st95} introduced 
an elegant and much needed design for writing
multi-threaded programs on shared memory. The 
promising feature in transactional memory proposal was to provide 
 simplified syntax that guarantees correctness, and 
 delegates the handling of mutual exclusion to the 
 library providing the implementation. Our approach and ideology is similar on these goals. However, constructs for 
 conditional waiting under transactional memory are limited
 \cite{no_wait_tm, lm11, ds09}. Hence,  
 writing many conditional synchronization based multi-threaded programs is rather difficult. 
 \remove{
 In addition, most of the 
 existing implementations \cite{tmimpl1,tmimpl2,tmimpl3}
 have not delivered on the performance side, at least not
 yet. For most 
 of the prevalent use cases, the existing 
 implementations, have inferior runtime performance, in 
 terms of overall throughput, in comparison to 
 the standard multi-threaded implementations that use
 conventional locks. The well known reason for the lack
 of performance is the lack of hardware support for 
 {\em atomic} block instructions on any platform, 
 and the resulting additional work transactional 
 memory implementations have to perform in 
 tracking the conflicting updates and rolling back 
 some transactions.} Also, unlike transactional memory, our approach 
 merely transfers the responsibility of data 
 manipulation to a single
 monitor thread and does not require any complicated rollback mechanism for resolving
 conflicting
 updates on the shared data.  

{Designing {\bf lock-free} and {\bf wait-free} algorithms 
is a challenging problem for any use-case. In addition, lock-free and wait-free implementations are currently available for a small set 
of data structures \cite{fhs04, her88}. Even though there is a lot of ongoing research to grow this set, the growth 
is slow. 
By using our approach, the end user 
benefits by an additional level of  abstraction and does 
not have to deal with the complications involved in understanding and correctly implementing
wait-free 
and lock-free algorithms. 

We use the concept of {\bf \em futures} \cite{future,lea05} to realize the 
idea of non-blocking/asynchronous executions.
The work by Kogan et al. \cite{kogan_future} is similar to our approach, and makes 
use of {\em futures} for non-blocking executions. However, 
there are few key differences between our work and
\cite{kogan_future}. We explore changes to the paradigm
of monitors, especially for conditional synchronization problems, whereas \cite{kogan_future} focuses on 
throughputs, under 
saturation scenarios, 
of data structures such as stacks, queues, 
and linked-lists. Hence, our approach provides a more
generic tool for performing safe updates and 
synchronization between threads. 

\remove{
issues\\
Lockfree and Waitfree - difficult to program, does 
not necessarily lead to faster performance\\
Java concurrent structures using lockfree and waitfree
algorithms\\
}

{\bf Frameworks} such as Legion \cite{legion} and Galois 
\cite{pnk+11} focus on improving runtime 
performance of parallel programs on scientific simulations and large graphs, respectively. Thus, they 
do not directly relate to thread synchronization and 
notification based use-cases such as bounded-buffer.
Secondly, 
both of these libraries require advanced skills 
and familiarity with the features and syntactic details
exposed by the libraries. Hence, the performance gains 
achieved by using them are highly 
correlated with a programmer's knowledge of the involved optimization parameters needed to tweak the executions for optimal performance. 

\remove{
have focused on improving parallel 
programs with designs that
 
Alex Aiken - Legion\\
Galois\\
AutoSynch and other parallel execution libraries/frameworks
}

%% file: conclusion.tex
\section{Conclusion} \label{sec:conclusion}
Proposing changes to a well established programming 
paradigm is an ambitious goal. Moreover, convincing 
 programmers to adopt a new programming framework in favor 
of the conventional one, usually requires the proposed 
approach to offer gains in terms 
of runtime performance as well as programming ease/productivity. The {\em ActiveMonitor} 
framework is a step in this direction. In almost all of the problems used in our evaluation, monitors with non-blocking 
executions outperform Java's reentrant lock based monitors. 
Additionally, this gain is achieved by a simple mechanism of using only two new keywords in 
existing Java programs. At present, our prototype implementation is costly in terms of processor 
usage. However, with the ubiquity of multi-core processors, optimized implementations 
of our proposed technique could lead to significant advantages at reduced costs. In this paper, our initial  
exploration in a new direction, that is: non-blocking executions on monitors, has showed promising results.
We hope that with further research 
in this direction, our approach would lead to even stronger results.

%% file: appendix.tex
\hfill\\
\begin{flushleft}
{\Large \bf Appendix}
\end{flushleft}
\appendix
\section{Implicit/Automatic signal monitors:}
\label{app:autosynch}
Monitors can be divided into 
two categories according to the different implementations of conditional 
synchronization: explicit-signal monitors and implicit-signal monitors \cite{bh05b}. In  explicit-signal monitors, condition
        variables are used along with {\em wait}, and {\em notfiy/notifyAll} statements for 
        conditional synchronization between 
        threads. 
    Programmers need to associate assertions with condition variables manually.
When a thread wants to execute a set of instructions 
under a critical section, it first checks 
some condition variable(s), and waits
    if the predicate is not true. Once a thread enters the waiting state, another thread 
on detecting that the conditional 
     predicate has become true, explicitly notifies
     the waiting thread. The bounded-buffer implementation 
     in Fig.~\ref{fig:bb_exp}(a) that uses Java's reentrant locks is an explicit-signal monitor.  In fact, 
almost all of the prevalent programming languages, including Java,  
use the explicit-signal design for monitors. 

Implicit-signal monitors, also called {\em automatic-signal} monitors, require the underlying system to handle
thread 
synchronization and wait/notification. Recently, Hung et al.  \cite{hg13} showed that implicit-signal monitors can be beneficial when implemented on 
modern multi-core processors. In an implicit-signal design, programmers need to 
    use {\sf \small waituntil}
    statements (line $9,15$ in automatic-signal program in
    Fig.~\ref{fig:bb_exp}(b)) instead of condition variables for
    synchronization. In this 
    monitor design, a thread will wait as long as the condition of a {\sf \small waituntil}
    statement is false, and execute the remaining instructions only after the condition 
    becomes true. The responsibility of identifying 
    a waiting thread to signal, and then 
    signaling it, is that of 
    the underlying system/framework rather 
    than of the program logic explicitly 
   put in the thread by a programmer. The details 
   of design and implementations of such monitors
   is beyond the scope of this paper, and we refer
   the interested reader to \cite{hoa74,hg13}.

\section{Monitor Task Examples}
\label{app:task_ex}
Here, we highlight some key aspects of our framework's handling of monitor tasks. First, note that the set of statements, 
$\mathcal{S}$, of a 
task may be empty. Such a task
serves as a barrier; and is legal in our framework. Second, the precondition can either be absent or appear later, and
not as
the first statement, in the monitor method.
When a monitor method has no precondition, our framework creates a
task with the precondition as tautology, indicating that the task can be
executed at any time. 
If a monitor method does not start with a {\sf \small waituntil} statement but has some 
such statement in between, then the precondition of the first derived task is a tautology. 
In Fig.~\ref{fig:method_ex}, the method 
{\tt bar()} has two corresponding tasks $T_1$ and $T_2$, where $T1$ has the 
precondition as tautology and the set of statements $S_1$; $T_2$ has the 
precondition $P$ and the set of statements $S_2$.

Furthermore, 
monitor tasks are compositional in nature. Consider  the method {\tt foo()} in Fig.~\ref{fig:method_ex}. If a method
declares $n$ sets of preconditions and statements, 
then the framework would generate $n$ tasks such that 
 each task $T_i$, $1 \le i \le n$, has 
a precondition $P_i$ and a corresponding set of statements $\mathcal{S}_i$. 
Appendix~\ref{app:task_ex} presents some detailed examples of these cases.

\begin{figure}[ht!]
\begin{multicols}{2}
    \begin{Verbatim}[fontsize=\footnotesize,gobble=6,
            codes={\catcode`$=3\catcode`_=8}]
        void foo() {
            $waituntil(P_1)$;
            $S_1$;
            $waituntil(P_2)$;
            $S_2$;
            ...
            $waituntil(P_n)$;
            $S_n$;
        }
    \end{Verbatim}
    \columnbreak 
    \begin{Verbatim}[fontsize=\footnotesize,gobble=2,
            codes={\catcode`$=3\catcode`_=8}]

        void bar() {
            $S_1$;
            $waituntil(P)$;
            $S_2$;
        }
    \end{Verbatim}
\end{multicols}
    \caption{Examples of methods leading to compositional tasks and tautology as precondition}
    \label{fig:method_ex}
\end{figure}

\section{Lemma Proofs}
\subsection{Proof Sketch for Lemma \ref{lem:lock}}
\label{app:proof_lock}

We show that for any execution in our model there exists 
an equivalent lock-based execution.
Since all tasks of any monitor object are executed by a single thread 
due to Rule~\ref{def:mutex}, mutual exclusion is preserved just as in any lock-based
execution. We only need to show that the order of execution of the tasks
corresponds to a schedule in which worker threads execute the tasks.

\remove{we show that the sequence of tasks executed corresponds to a
sequence possible for a lock-based execution.}
It is sufficient to show that all tasks submitted by a single worker
thread execute in the order of submissions.
Let $s$ and $t$ be two consecutive tasks submitted by a worker thread.
If they are submitted for the same monitor, then the Rule~\ref{rule:sc} preserves the order. 
If $s$ is a blocking task, then by definition of blocking task, 
$t$ cannot be submitted before $s$ 
is completed. Hence, execution of $s$ precedes execution of $t$.
If $s$ is a non-blocking task and is on a different monitor object from $t$,
then due to Rule~\ref{rule:linear_active1} we wait for $s$ to finish before submission of $t$.

\subsection{Proof Sketch for Lemma \ref{lem:linearizable}}
\label{app:proof_linearizable}
It is sufficient to show that
for every thread history there exists an equivalent
sequential thread history that is consistent with the execution by the monitor thread.
We get this sequence by considering as linearization point
for an operation the instant at which the monitor thread finishes
executing the corresponding task. We show that this order
is consistent with the thread order.

Let $s < t$ denote that operation(i.e. its corresponding task) $s$ was executed before $t$. 
First consider the case when $s < t$ because $s$ is a blocking operation.
In this case, task $t$ cannot be submitted before task $s$
is executed. Hence, the order of execution of monitor tasks preserves the thread
order. Second, consider the case
when $s < t$ because they are operations on the same monitor object and $s$ is
invoked before $t$. In this case Rule~\ref{rule:sc} guarantees that the order of execution
is $s$ followed by $t$.
Finally, consider the case when $s$ and $t$ are tasks submitted by 
the same thread, $s$ is non-blocking but its results are used by $t$.
In this case, the execution would have to block for collecting the results of $s$ before
executing any further. Hence $t$ would be submitted later than the execution 
of $s$.

\section{Implementation Details}
\label{app:impl}
The translation of the code written with our two proposed keywords to an equivalent Java code is performed  by
using a pre-processor. For our implementation, we use JavaCC \cite{kod04} pre-processor. 
We convert all the monitor 
methods to equivalent tasks, and use a threadpool based executor service framework for monitor threads. 
Here, we first discuss the pre-processing steps.

\subsection{Pre-processing}
\label{sec:app_preproc}
We briefly summarize the concepts applied for compiling the code, using the {\em ActiveMonitor} framework, 
in which the two proposed keywords are used. 

First, our pre-processor identifies the {\sf \small waituntil} statements that
capture the preconditions of monitor tasks. Recall that the {\sf \small
waituntil}  keyword is used for automatic signaling between threads, and the
predicate provided as its argument forms the precondition for execution of
monitor tasks. In the corresponding generated code (in Java) every predicate is
created as an inner class with a method called $isTrue()$, which returns a boolean value 
as the 
result of the predicate evaluation. At runtime, our system invokes the $isTrue()$ method to evaluate  the
predicate when deciding which task should be executed. 

Next, the pre-processor creates a {\sf \small Callable} \cite{lea05} or {\sf \small Runnable} \cite{lea05}
object that contains the set of statements for each task inside its $call()$ or
$run()$ method. If the return value of the original monitor method is void ,
then we create a {\sf \small Runnable} object; otherwise, we create a {\sf
\small Callable} object. Note
that if a non-blocking monitor method has been divided into multiple tasks, the
resulting tasks may have some shared variables. To handle these variables, our
system creates a data object to store them and all the tasks can access them
through this data object.

Finally, we replace the invocations of monitor methods by submitting monitor
tasks to threadpool executor (Appendix~D.2). The executor returns {\sf \small f},
 an instance of Java's {\sf \small Future} object \cite{lea05},  to the invoker thread when a task is submitted. If the task is
blocking, then the thread needs to wait for the computation in {\sf \small f} to finish;
this is done by evaluating {\sf \small f} within the task and returning the result.
Otherwise, {\sf \small f} is registered with the threadpool executor for exception checking
and handling.  

In Section~\ref{sec:concept}, we mentioned that our current implementation does not admit a particular family of methods for 
task based executions. Specifically, it does not admit recursive blocking calls, monitor methods in which conditional synchronization guards using automatic thread signaling ({\sf \small waituntil}) are not initial statements of the program, or are present in conditional branches. This does not 
mean that programs requiring such method implementations are not allowed in our framework; just that at present we do not support
task generation of such methods. A programmer can implement them using conventional Java syntax, and 
still use our framework for other methods. 

\subsubsection{Discussion: Thread Dependent Variables and Functions}
In our current implementation, 
thread dependent variables and functions within a monitor method cannot be used
directly in the {\sf \small Runnable} or {\sf \small Callable} object that is used in task generation
by our approach. This is because the tasks are executed by
the monitor thread and not by the worker thread. For example, suppose there is a 
monitor method that invokes {\sf \small Thread.currentThread()}, if we directly
add this statement to the generated {\sf \small Runnable} object (in the task), then this method's invocation
at runtime will return the 
reference to the monitor thread
when it is executed. However, it is obvious that the intent of this call inside the monitor method 
was to refer to the worker thread. To handle this situation, currently, we require the programmer 
to perform reference copy and storage and storage in thread-local variables. 
For read operations of thread dependent variables and functions, the worker thread
would need to evaluate them outside the monitor, and store the result with final variables. These final variables
can be accessed by the runnable and callable objects. 
An additional constraint/limitation applies for the case of write operation on thread dependent 
variables. For write operations, if the monitor method is non-blocking then the
results can be stored as intermediate data. The worker thread then writes these
results back to its local variable after the task is executed. 

\subsubsection{Discussion: Blocking recursive method}

Our current pre-processing implementation does not support a {\emph blocking} recursive method on monitors.
This is because the number
of the method invocations to be made at the runtime is non-deterministic. Thus, we cannot know how many tasks we
need to create at pre-processing time. In addition, since the method is
blocking, the monitor thread will get blocked when it recurs. 

\subsection{Limiting Monitor Threads using ThreadPools} \label{sec:threadpool} 
A threadpool is
a collection of threads to execute tasks that are units of computation. These
threads are usually created together at the start-up , and remain in the pool to
provide executor service. For  task execution based programs, the use of a
threadpool  can significantly improve the runtime performance by having an
already existing thread ready to execute the tasks as they arrive.  Generally,
tasks are stored in a collection, and free threads are responsible for finishing
the unexecuted tasks.  If there is no task that is eligible for execution, the
threads in the pool wait for one to arrive.  This approach is especially
beneficial when the number of tasks is greater than the number of threads in a
pool. Thus, the size of the threadpool is a crucial factor as the timing of
creating or destroying a thread may have a significant impact on performance.
Given that we spawn a new thread for a monitor object, programs with a large
number of monitors could be adversely affected by having too many additional
threads in the {\em ActiveMonitor} framework.

To keep the 
overhead of monitor threads relatively low,
we control the number of monitor threads created by querying the operating system 
for available resources and managing the monitor threads dynamically. A component, called {\em Monitor ThreadPool Executor}, 
manages this process. 
This component  instantiates monitor threads as threads of
a threadpool
at the start of the execution.  
The pre-processing phase collects the information about number of 
monitor objects in the program, and the threadpool executor uses this information for creating a threadpool balancing
the resource consumption and performance benefit. 
 
\subsubsection{Exception Handling} \label{app:exception} For a non-blocking
method invocation, after submitting its corresponding task to the executor, the
invoker does not need to wait for the completion of the task. The task is
executed in parallel by a thread of the monitor threadpool. Thus, if an
exception occurs during its execution, the thread that submitted it must be
notified of this exception.  The {\em ActiveMonitor} framework has an exception
handler that keeps a log of every exception and provides different mechanisms
for programmers to handle exceptions in the non-blocking method. The users may
choose to ignore the exceptions or they can specify a maximum number of times a
task may be considered for   automatic re-tries. Furthermore, our system also
provides a hook so that the programmer can write their custom exception
handler.